\documentclass{acmart}

\AtBeginDocument{%
  \providecommand\BibTeX{{%
    \normalfont B\kern-0.5em{\scshape i\kern-0.25em b}\kern-0.8em\TeX}}}

\setcopyright{acmcopyright}
\copyrightyear{2022}
\acmYear{2022}
\acmDOI{XXXXXXX.XXXXXXX}

\acmConference[IWCTS '22]{}{}{}
\acmBooktitle{} 
\acmPrice{}
\acmISBN{}
\usepackage{lipsum}
\usepackage{xcolor}
\usepackage{listings}
\usepackage{svg}
\usepackage{amsmath}

\usepackage{xpatch}

\usepackage{siunitx}
\definecolor{keyword}{HTML}{2771a3}
\definecolor{pattern}{HTML}{b53c2f}
\definecolor{string}{HTML}{be681c}
\definecolor{relation}{HTML}{7e4894}
\definecolor{variable}{HTML}{107762}
\definecolor{comment}{HTML}{8d9094}

\usepackage[normalem]{ulem}
\newcommand{\keepcomment}{1} 
\ifnum\keepcomment=1
	\usepackage[colorinlistoftodos,textsize=scriptsize]{todonotes} 
    \newcommand{\stkout}[1]{\ifmmode\text{\sout{\ensuremath{#1}}}\else\sout{#1}\fi}
    
    \newcommand{\todoi}[1]{\todo[inline]{#1}}
\else
	
	\usepackage[disable]{todonotes} 
	\newcommand{\todoi}[1]{\leavevmode\ignorespaces\unskip}
\fi

\begin{document}

\title[A graph-database approach to assess the impact of demand-responsive services on public transit accessibility]{A graph-database approach to assess the impact of demand-responsive services on public transit accessibility}

\settopmatter{authorsperrow=4}

\author{Cathia Le Hasif}
\email{cathia.lehasif@gmail.com}
\affiliation{%
 \institution{Université d'Évry}
\country{France}
}

\author{Andrea Araldo}
\email{araldo@imtbs-tsp.eu}
\affiliation{%
  \institution{Télécom SudParis, IP Paris}
  \country{France}
}

\author{Stefania Dumbrava}
\email{stefania.dumbrava@ensiie.fr}
\affiliation{%
  \institution{ENSIIE, IP Paris}
  \country{France}
}

\author{Dimitri Watel}
\email{dimitri.watel@ensiie.fr}
\affiliation{%
  \institution{ENSIIE, IP Paris}
  \country{France}
}

\renewcommand{\shortauthors}{Araldo et al.}

\lstdefinelanguage{cypher}
{
	morekeywords={
		MATCH, OPTIONAL, WHERE, NOT, AND, OR, XOR, RETURN, DISTINCT, ORDER, BY, ASC, ASCENDING, DESC, DESCENDING, UNWIND, AS, UNION, WITH, ALL, CREATE, CALL, YIELD, DELETE, DETACH, REMOVE, SET, MERGE, SET, SKIP, LIMIT, IN, ON, CASE, WHEN, THEN, ELSE, END, INDEX, DROP, UNIQUE, CONSTRAINT, EXPLAIN, PROFILE, START,
	}
}

\newcommand{\mycdots}{\cdot\!\cdot\!\cdot}
\lstset{language=cypher,
	literate=*
	{...}{$\mycdots$}{1}
	{theta}{$\theta$}{1}
}

\begin{abstract}
This paper proposes an approach to analyze the impact of multimodal Public Transit (PT), combining conventional fixed-route transit and Demand-Responsive Transit (DRT), on equality in transport accessibility distribution. 
We construct a graph model of multimodal PT in Neo4j, based on General Transit Feed Specification (GTFS) data, modeling DRT analytically (using continuous approximation). We quantify service quality on any location using an \emph{accessibility} measure, indicating how easily other places or opportunities can be reached. We quantify inequality of accessibility distribution. We state the problem of allocating a DRT fleet to minimize inequality and show its NP-completeness. We showcase our approach on the transportation network of the French town of Royan.

\end{abstract}



\maketitle

\section{Introduction}
Effective mobility shapes the quality of life and enables economic development of urban conurbations. At the same time, transport is one of the most pollutant human activities and urgent actions are needed in order to achieve the ambitious zero-emission goals that several (supra)national authorities set for the following years. 
While Public Transit (PT) plays a pivotal role to guarantee sustainable and effective mobility in urban conurbations, the quality of service is unequally distributed~\cite{Badeanlou2022}, keeping a large part of the population with poor PT service, thus forcing them to rely on private cars. Complementing conventional PT, solely based on fixed lines, with other more flexible modes can relieve this issue~\cite{araldo2022pooling,Calabro2021,Araldo2019a}. In this paper, we consider a Demand-Responsive Transit (DRT) service, in which vehicles adapt their routes to user requests. Despite the fact that DRT has been shown to be an effective complement of PT in low-demand areas~\cite{QUADRIFOGLIO2009,Calabro2021}, no studies exist, to the best of our knowledge, to quantify its impact in the reduction of inequality of PT accessibility. \emph{Ours is a first preliminary work toward an open data-driven computational approach to enable such quantification}. 

To efficiently and scalably analyze PT, we resort to graph databases systems, which have been developed to address the ever-growing need for processing highly interconnected data, in a wide variety of areas. We leverage these to implement our formal model, to visualize and explore transport graphs, as well as to extract patterns and compute relevant metrics, using custom graph queries and algorithms. We summarize our contributions as follows:


\begin{itemize}
    \item We formalize an optimization for allocating DRT vehicles with the objective of reducing  accessibility distribution inequality and we show its NP-completeness.
    \item We model multi-modal transit, consisting of conventional fixed lines and DRT (represented analytically via continuous approximation). The model is a single graph, amenable to simple computation of accessibility computation, via well-known graph-algorithms.
    \item We implement the above model in the Neo4j graph database \cite{neo4j:website} as a property graph, built based on open data (GTFS) \cite{gtfs:website}. 
\end{itemize}
Note that \emph{our proposed approach is generic} and can be used to analyze any city for which the corresponding GTFS data is available. For illustrative purposes, we have chosen the city of Royan as use-case. The size of its PT graph is large enough to be representative, while maintaining tractability ($\approx$ 18.8K nodes and $\approx$ 7.3K arcs) and its structure has all the elements needed to observe the impact that deploying DRT services has on the accessibility distribution.
While the paper presents preliminary results, we are working on extending it into a usable tool for PT planners (\S\ref{sec:conclusion}).


\section{Model}
\label{sec:model}


Table~\ref{tab:notation} provides the key notations and values used henceforth.

\begin{table}[]
    \centering
    \begin{scriptsize}
    \setlength{\tabcolsep}{2.5pt}
    \begin{tabular}{|l|l|l|l|}
    \hline
        \textbf{Symbol} & \textbf{Description} &
        \textbf{Values}
    \\
    \hline
        $\mathcal{V},\mathcal{A},\omega(\cdot)$ & 
        Nodes, arcs, weights
        & $|\mathcal{V}|$=18892, $|\mathcal{A}|$=73714, $\omega(\cdot) \in [0,38400]$
    \\
    \hline
        $\xi$ & Tessellation size & 1km
    \\
    \hline
        $\tau$ & Walk distance (all walk arcs$\le\tau$ )
        & 1km
    \\
    \hline
        $v^w$ & Walk speed & 3km/h
    \\
    \hline
        $L,W$ & Height and width of the DRT area
        & L = 4km, W = 2km
    \\
    \hline
        $\mathcal{S}, \mathcal{S}^\text{DRT}$
        & Stops [served by DRT] & $\mathcal{S}$ = 749, $\mathcal{S}^\text{DRT}$ = 20
    \\
    \hline
        $\mathcal{C}, \mathcal{C}(s)$
        & Centroids [in the DRT region of $s$] & $\mathcal{C}$ = 1320, $\mathcal{C}(s)$ = 16
    \\
    \hline
        $\mathcal{D}(c)$ & Potential destinations of $c$
        &$\left\{ s\in\mathcal{S} | 2\text{km} \le d(c,s) \le 10\text{km} \right\}$
    \\
    \hline
        $\rho$ & Demand density & $26$ pax/km$^2$-h
    \\
    \hline
        $t_0$ & Considered departure time
        & 8h30
    \\
    \hline
        $M$ & Avg. no. of needed DRT vehicles
        & $M=21.0$
    \\
    \hline
    \end{tabular}
    \setlength{\belowcaptionskip}{-2.5pt}
    \caption{Notation and numerical values}
    \label{tab:notation}
    \end{scriptsize}
\vspace{-3.2em}
\end{table}

We compute accessibility on a generic weighted directed graph $G=(\mathcal{V},\mathcal{A},\omega(\cdot))$, where $\mathcal{V}$ is the set of nodes, $\mathcal{A}$ - the set of arcs, and $\omega(\cdot)$ - the weight function. We partition $\mathcal{V}$ into: a set $\mathcal{C}$ of \emph{centroids}, from which we assume all passengers start, a set $\mathcal{S}$ of physical stops and, for each $s \in \mathcal{S}$, a set $\mathcal{ST}(s)$ corresponding to the service of the stop, from a specific transit line, at a specific time (called hereinafter a \emph{stoptime}). Further details are given in \S\ref{sec:graph}. 


Within the set $\mathcal{A}$, we distinguish the arcs $\mathcal{A}^\text{w}$, which represent a walk connection between any centroid $c$ and a close stoptime $v$. 

In addition, a subset of stops $\mathcal{S}^\text{DRT}$ can be served by a DRT service deployed in a limited region close to each stop. We call $\mathcal{C}(s)$ the set of centroids within that region, for any $s\in\mathcal{S}^\text{DRT}$. The set of arcs $\mathcal{A}^\text{DRT} \subset \mathcal{A}$ represent DRT connections between each centroid in $\mathcal{C}(s)$ and stoptimes $v\in\mathcal{ST}(s)$. For any DRT-arc $a\in\mathcal{A}^\text{DRT}$, the weight $\omega(a)$ depends on the number $q$ of DRT vehicles deployed around the corresponding stop.
The formulas for computing the graph's weights are described in \S\ref{sec:graph}.

\subsection{Accessibility and inequality}
\label{sec:accessibility}
We evaluate the inequality in the distribution of PT \emph{accessibility}~\cite{Miller2019}. The accessibility of a location indicates how easy it is to reach other locations. For ease of representation, we focus on its opposite, i.e., the \emph{inaccessibility}. Let $T(c, v)$ be the weight of the shortest path from a centroid $c$ to a stop in $\mathcal{D}(c)\subseteq\mathcal{S}$ - a set of potential destinations (Table~\ref{tab:notation}). The inaccessibility of the centroid $c \in \mathcal{C}$ is 
$
ia(c) = \frac{1}{|\mathcal{D}(c)|}\cdot\left( \sum\limits_{s \in \mathcal{D}(c)} T(c, s)\right)
$.
To quantify inequality, we compute the Palma index~\cite{palma2011homogeneous}. We denote with $\mathcal{C}_{10}^\text{worst}$ the set of 10\% of centroids suffering from the largest inaccessibility and with $\mathcal{C}_{40}^\text{best}$ the set of 40\% of centroids enjoying the lowest inaccessibility, we compute inequity as 
$
  I(G, \{\omega(a)\}_{a\in\mathcal{A}}) = 
  \sum\limits_{c \in \mathcal{C}_{10}^\text{worst} } ia(c) / \sum\limits_{c \in \mathcal{C}_{40}^\text{best} } ia(c)
$.

This value can be obtained in polynomial time, by adapting the Dijkstra algorithm to compute $T(c, v)$, for each pair of centroid and stop. 
Note that the more the accessibility values are equal, the less $I(G, \omega)$ is. Our objective is then to minimize the value of $I(G, \omega)$, without diminishing the accessibility of any centroid. To do so, we have one lever: we are given $Q$ routing vehicles (DRT buses, for instance) that can be used to decrease the time needed to reach stop $s$ from a centroid. 

\subsection{Problem Statement}
We now state the problem of assigning DRT vehicles to different regions, in order to minimize inequality (\S\ref{sec:accessibility}), and show its NP-completeness. We will explore resolution heuristics in future work.

Let $f: \mathcal{A}^{DRT} \times \mathbb{N} \rightarrow \mathbb{N}$ be a function that can be computed in polynomial time. By allocating $q$ DRT vehicles to the stop $s \in \mathcal{S}^{DRT}$, the weight of every arc $(c, v) \in \mathcal{A}^{DRT}$, from any centroid $c \in\mathcal{C}(s)$ to any stoptime $\mathcal{ST}(s)$, is reduced by $f((c,v), q)$. This increases the accessibility of those centroids. We define the assignment vector $\mathbf{q}=(q_s, s \in \mathcal{S}^\text{DRT})\in\mathbb{N}^{|\mathcal{S}^\text{DRT}|}$,
which allocates $q_s$ vehicles to each $s \in \mathcal{S}^\text{DRT}$.
Let $\{\omega_\mathbf{q}(a)\}_{a\in\mathcal{A}}$ be the weights obtained from $\mathbf{q}$. Our optimization problem can thus be formulated as follows.

\newcommand{\inequality}{I\left(G, \{\omega_\mathbf{q}(a)\}_{a\in\mathcal{A}} \right)}

\begin{definition}[Problem (EQ)]
\label{def:problem}
Given a graph $G = (\mathcal{V}, \mathcal{A})$, a weight function $\omega$ and an integer $Q$, find the assignment $\mathbf{q}=(q_s, s \in \mathcal{S}^\text{DRT})$ satisfying$\sum\limits_{s \in \mathcal{S}^\text{DRT}} q_s \leq Q$ and minimizing 
$\inequality$.
\end{definition}

\paragraph{NP-Completeness} We now prove the NP-Completeness of this problem in the general case, using a polynomial reduction from the knapsack problem, in which, given two sets $(d_1, d_2, \dots, d_n)$ and $(u_1, u_2, \dots, u_n)$ of integers and two integers $D$ and $U$, we search for $J \subset \llbracket 1; n \rrbracket$, such that $\sum_{i \in J} u_i \geq U$ and $\sum_{i \in J}d_i \leq D$. 
\begin{theorem}
    The decision version of (EQ) is NP-Complete.
\end{theorem}
\begin{proof}
We associate a rational $IN$ to the problem (EQ) and search for the existence of a feasible assignment $\mathbf{q}$ with inequality $\inequality$ lower than $IN$. The inclusion in NP is trivial. Considering that $f$ can be computed in polynomial time, one can compute $\inequality$ in polynomial time, for any given assignment $\mathbf{q}$.

We now prove there exists a polynomial time reduction from the knapsack problem to (EQ). From a knapsack instance $\mathcal{J}$, we build an instance $\mathcal{I} = (G, \omega, f, Q)$ as follows:


\hspace{-1em}
\fbox {
\hspace{-2em}
\parbox{\linewidth}{
\begin{small}
\begin{itemize}
\item We set an integer $M$ to $\max(u_i + 2)$.
\item We construct the graph $G$ containing $n+1$ stops $s_1, s_2, \dots, s_{n}, s_{n+1}$, $n + 1$ stoptimes $v_1, v_2, \dots, v_{n+1}$ and $10n$ centroids $c_1, c_2, \dots, c_{n}, c'_{n+1}, c'_{n+2}, \dots, c'_{10n}$. 
\item We set $\mathcal{C}(s_i) = \{c_i\}$, $\mathcal{ST}(s_i) = \{v_i\}$, and $\mathcal{S}^{DRT} = \{s_1, s_2, \dots, s_n\}$.
\item For $i \in \llbracket 1; n \rrbracket$, each centroid $c_i$ is linked to the stoptime $v_i$ with one arc $(c_i, v_i)^w \in \mathcal{A}^w$ and another arc $(c_i, v_i)^{DRT} \in \mathcal{A}^{DRT}$. For $i \in \llbracket n+1; 10n \rrbracket$, each centroid $c'_i$ is linked to the stoptime $v_{n+1}$ with only one arc $(c'_i, v_{n+1})^w \in \mathcal{A}^w$. In addition, $\mathcal{A}$ contains the arcs $(v_i, s_i)$, $(v_i, v_{n + 1})$ and $(v_{n + 1}, v_i)$, for $i \neq n + 1$.
\item For $i \in \llbracket 1; n \rrbracket$, $\omega((c_i, v_i)^w) = \omega((c_i, v_i)^{DRT}) = M$. For $i \in \llbracket n+1; 10n \rrbracket$, $\omega((c'_i, v_{n+1})^w) = 1$. Other arcs are weighted with 0.
\item For $i \in \llbracket 1; n \rrbracket$, $f((c_i, v_i), q) = 0$ if $q < d_i$, and $M - u_i$ otherwise.
\item Finally, we set $Q = D$ and $IN = \frac{nM - U}{4n}$.
\end{itemize}
\end{small}
}
}
\vspace{0.5em}

This transformation from $\mathcal{J}$ to $\mathcal{I}$ is polynomial. 

Note that we can consider only the feasible assignments $\mathbf{q}$ of $\mathcal{I}$, where $q_{s_i}$ is either 0 or $d_i$, for every $i \leq n$. We call such a solution a canonical assignment. For any non canonical assignment, there exists a canonical one with same Palma index, since by construction $f((c_i, v_i), q)$ is null on $\llbracket 0; d_i - 1 \rrbracket$ and constant on $\llbracket d_i; Q \rrbracket$. In addition, there exists a bijection from the canonical assignments to the subsets of $\llbracket 1; n \rrbracket$ satisfying $\sum_{i \in J} d_i \leq D$: from $\mathbf{q}$, we can define $J \subset \llbracket 1; n \rrbracket$ as the set $\{i | q_{v_i} \geq d_i \}$. Thus, the associated canonical assignment has the inaccessibility values: $ia(c'_{i}) = 1$ if $i \geq n + 1$, $ia(c_i) = M$, if $i \not\in J$, and $ia(c_i) = M - u_i$, otherwise. Note that the two last values are greater than 2: $\mathcal{C}_{10}^\text{worst} = \{c_1, c_2, \dots, c_n\}$ and $\mathcal{C}_{40}^\text{best} = \{c'_{6n+1}, c'_{6n+2}, \dots, c'_{10n}\}$. Thus, $\inequality = nM - \sum_{i \in J}u_i / 4n$. We then immediately have that $\sum_{j \in J} u_i \geq U$ if and only if $\inequality \leq IN$.

Then, there exists a polynomial time reduction from the knapsack problem to (EQ).
\end{proof} 
Note that this reduction could be adapted to any fixed network. Changing the network would add a fixed value to the accessibility of all the centroids that can be balanced with the weights of the arcs and the function $f$. However this idea is not immediately achievable if we consider realistic functions $f$ where the value $f(v, q)$ should increase continuously with $q$, possibly with a threshold representing the number of vehicles needed before the driving time takes an advantage on the walking time. Further work is then needed to search for the complexity of this problem in realistic instances. 

\subsection{Graph model of multimodal public transit}
\label{sec:graph}

\subsubsection{Graph model of fixed public transit}
We now render the generic graph model from \S\ref{sec:model} more specific. We detail the construction of the arcs and weights and how such concepts map to GTFS data. Let us denote with $\mathcal{L}$ and with $\mathcal{R}_l$ the set of \emph{runs}, each corresponding to a vehicle traveling along the line and departing from the terminal at a specific time. GTFS data associate a stoptime $v=(s,r)$ the event of a run $r$, of a line $l$, passing by physical stop $s$. We denote with $t(v)$ the time in which this even happens. In our graph model stoptimes are represented as nodes. The stops times related to a physical stop and line $l$ serving that stop are 
\begin{align}
\nonumber
\mathcal{ST}(s,l) &=\{v=(s,r)|r\in\mathcal{R}^l.\} 
\\
\nonumber
\mathcal{ST}(s) &=\bigcup_{l\in\mathcal{L}} \mathcal{ST}(s,l).
\end{align}

We add arcs $a$ between stoptimes $v$ and $v'$ based on the following rules. An arc $a$ is added between stoptime $v$ and $v'$ if there exists a run (of any line) that serves $v$ and then $v'$ immediately after. Take a stoptime $v\in\mathcal{ST}(s,l)$. We add the following set of arcs, that allow to change from a line to another at stop $s$:
\begin{equation}
\label{eq:stoptime-arcs}
\left\{
    a=(v,v_{l'}) | l'\in\mathcal{L}\text{ serving }s, v_{l'}\in\mathcal{ST}(s,l'), v_{l'} = \underset{v'\in\mathcal{ST}(s,l')}{\arg\min} t(v')
\right\}.
\end{equation}

To let paths end to physical stops, we add the following arcs:
\begin{equation}
\label{eq:zero-arcs}
    \left\{
    a=(v,s) | v\in\mathcal{ST}(s), s\in\mathcal{S}
\right\}
\end{equation}

To define the centroids $\mathcal{C}$, we tessellate the area into square tiles of side $\xi$. The centroids are the centers of such tiles. We consider journeys start at time $t_0$. The time at which a user can reach stop $s$ is $t^w_{c,s}=t_0+d(c,s)/v^w$, where $v^w$ is the walking speed and $d(c,s)$ is the distance between centroid $c$ and stop $s$. We connect each centroid $c$ with \emph{walk arcs} to stoptimes within walk distance $\tau$ that can be reached before the time of service:
\begin{align}
\label{eq:walk-arcs}
\left\{
\begin{array}{ll}
    a=(c,v)^w |  &
     l\in\mathcal{L}\text{ serves }s, v\in\mathcal{ST}(s,l), \\
     & 
     v = \underset{v'\in\mathcal{ST}(s,l), t_{c,s}^w\le t(v')}{\arg\min} t(v'),
        s\in\mathcal{S}, d(c,s)\le \tau
\end{array}
\right\},
\end{align}
We compute arc weights as follows. The arcs defined in~\eqref{eq:zero-arcs} have weight $0$. Arcs $a=(c,v)^w$ defined in~\eqref{eq:zero-arcs} have weight $w(a)=t(v)-t_0$. Arcs $a=(v,v_{l'})$ defined in~\eqref{eq:stoptime-arcs} have weight $w(a)=t(v_{l'})-t(v)$.

\subsubsection{Model of Demand-Responsive service}
When deploying the DRT feeder bus for stop $s\in\mathcal{S}^\text{DRT}$, such service spans a rectangular area of size $L\times W$ on the right of stop $s$, as in~\cite{QUADRIFOGLIO2009}. $\mathcal{C}(s)$ is the set of centroids within such an area. The DRT is $\emph{cyclic}$: during each cycle, a bus starts from $s$, picks-up/drops-off passengers in/to the area and returns to $s$. It thus stops for 2 minutes before starting the next cycle. A new bus departs every $h(s)$ units of time (headway). We apply the continuous approximation model of~\cite{Calabro2021}, of which we give here only minimum information, for lack of space.  The expected time needed to complete a cycle is $T(s)$, which depends on the number of passengers departing from/arriving to the region. We assume demand density $\rho$ (passengers willing to be picked-up - we assume the same density of passengers to be dropped-off) is uniform across the entire urban area. The higher $\rho$, the higher the detours imposed to the DRT, the higher $T(s)$. On average, the time needed to go from a centroid to stop $s$ via DRT is $T(s)/2$. Therefore, a passenger can arrive at stop $s$ at time $t_{c,s}^\text{DRT}=t_0+T(s)/2$. We thus add the following \emph{DRT arcs}:
\begin{small}
\begin{align}
\nonumber
\left\{
\begin{array}{ll}
    a=(c,v)^\text{DRT} |  &
     l\in\mathcal{L}\text{ serves }s, v\in\mathcal{ST}(s,l), \\
     & 
     v = \underset{v'\in\mathcal{ST}(s,l), t_{c,s}^\text{DRT}\le t(v')}{\arg\min} t(v'),
        s\in\mathcal{S}^\text{DRT}, c\in\mathcal{C}(s)
\end{array}
\right\}
\end{align}
\end{small}

As usual, the weight on such arcs $a=(c,v)^\text{DRT}$ are $t(v)-t(c)$. Imposing a certain headway $h$ requires to deploy a certain number of vehicles $M_h$ (see~\cite{Calabro2021} for the formula).

\subsubsection{Shortest path}
Note that after the construction explained above, graph $G$ models multimodal PT, as it includes both fixed PT and DRT services. We can compute shortest path $\hat T(c,s)$ from any centroid $c$ to any stop $s$, using Dijkstra algorithm or others. In case no path exist, we set $\hat T(c,s)=\infty$. To compute accessibility (\S\ref{sec:accessibility}), we use
$T(c,s)=\min\{\hat T(c,s), d(c,s)/v^w\}$, i.e., we do not consider PT paths that would be longer than walking.

\subsection{Graph Database Model}

To store and process transportation data, we leverage the leading open-source Neo4j graph database. Its \emph{property graph} model consists of a labeled, directed multigraph, in which lists of properties (key-value pairs) can be attached to both nodes and edges. For querying, its Cypher language allows to extract label-constrained reachability information encoding complex graph patterns \cite{BD18}. The Neo4j Graph Data Science (GDS) library contains efficient implementations of common graph algorithms, useful for analytics.

We build our transportation graph model using GTFS data, a standard open data format detailing transit schedules. Specifically, we import the information corresponding to the Royan urban region in Neo4j, as part of a custom graph instance we construct. 

We model the urban transport network as a property graph instance (see Figure \ref{fig:g1instance}). \texttt{Centroid} nodes (tessellation centroids) are linked to \texttt{Stoptime} nodes, which represent scheduled passage times and correspond to specific \texttt{Stop} nodes (stations), as indicated by \texttt{Located\_at} arcs. The arcs connecting a given \texttt{Centroid} and \texttt{Stoptime} are labeled \texttt{Walk} and/or \texttt{DRT}, depending on the pedestrian and, respectively, demand-responsive accessibility. \texttt{Stoptime} nodes are connected with \texttt{Precedes}-labeled arcs (temporal ordering) or with \texttt{Correspondance}-labeled arcs (correspondance, using a different line, between a pair of stoptimes of the same stop) 
The full graph instance amounts to 18892 nodes and 73714 relations (arcs), additionally annotated with 16 different node properties and 6 different relation properties. Note that the node values in Figure \ref{fig:g1instance}-right correspond to the identifiers of the corresponding \texttt{Stop} (also propagated to its stoptimes) or \texttt{Centroid} node.
                
\begin{figure}
    \centering
    \includegraphics[width=0.6\linewidth]{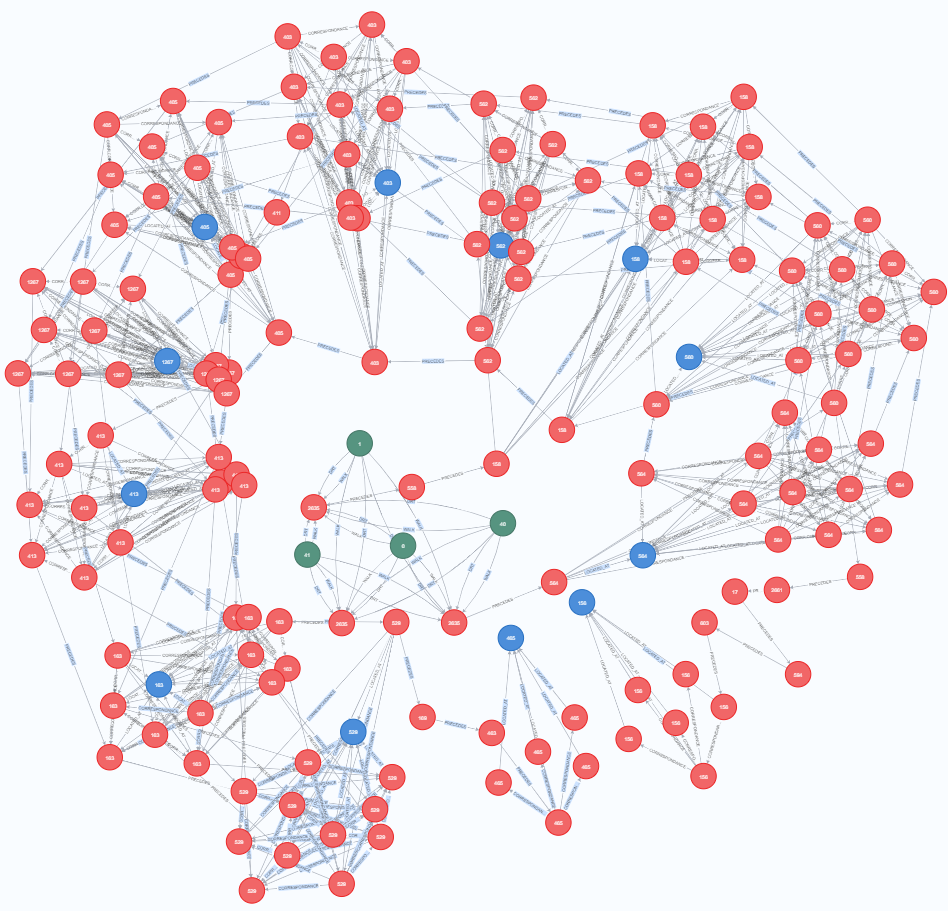}
    \includegraphics[width=0.6\linewidth,angle=90]{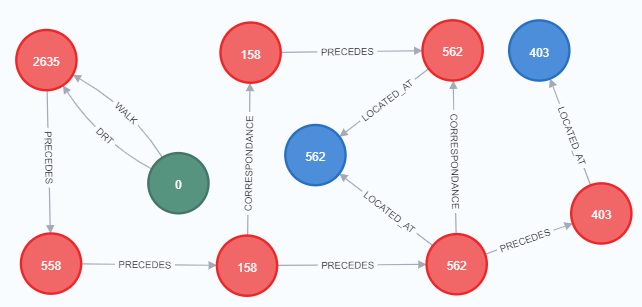}
    \setlength{\belowcaptionskip}{-15pt}
    \caption{Left: Neo4j snapshot of a portion of Royan's PT graph instance; the node colors encode their label: green (\texttt{Centroid}), red (\texttt{Stoptime}), and blue (\texttt{Stop}). Right: Detail.}
    \label{fig:g1instance}
\end{figure}

For our analysis, we enrich the property graph instance with accessibility information, stored using custom edge properties. We illustrate an example of this with the Cypher query below, in which we compute the shortest path between a given centroid and a stop. This uses the built-in Dijkstra's source-target shortest path algorithm from Neo4j's GDS library, weighted by the \texttt{inter\_time}, representing the walk time from a centroid to a station and the additional wait time at the station. The following example is a query that computes the shortest path from centroid $c=0$ to stop $s=607$.


\begin{lstlisting} [ basicstyle=\footnotesize,
label={lst:cypher},language=cypher]
MATCH (s:Centroid), (t:Stop) WHERE s.centroid_id = 0 AND t.stop_id = '607'
CALL gds.shortestPath.dijkstra.stream ('graph', {sourceNode: id(s),
        targetNode: id(t), relationshipWeightProperty: 'inter_time'})
YIELD sourceNode, targetNode, totalCost, nodeIds, path
CALL apoc.algo.cover(nodeIds) YIELD rel 
WITH startNode(rel) AS a, endNode(rel) AS b, rel AS r, 
     path AS p, totalCost AS tCost
RETURN p, a.stop_id AS from, b.stop_id AS to, r.inter_time AS inter_times,
r.walking_time AS walking_time, r.waiting_time AS  waiting_DRT,
r.travel_time AS DRT_time, tCost AS totaltime
\end{lstlisting}

\vspace{-0.2cm}

\section{Numerical results}
We showcase our approach on the French town of Royan, counting 749 PT bus stops. We select 20 $\mathcal{S}^\text{DRT}$ stops to be served by DRT (bold ones in Figure~\ref{fig:carte}-left). The considered numerical values are in Table~\ref{tab:notation}. 
We measure the reduction in the inaccessibility of centroids when introducing the DRT service, with headway $h=4$ min (resulting in 21.4 vehicles and 27.7 passengers per vehicle).
We plot in Figure~\ref{fig:carte}-right the reduction of inaccessibility observed by the centroids in $\mathcal{C}(s),\forall s\in\mathcal{S}^{DRT}$: only some centroids benefit from DRT, which suggests that it is of capital importance to solve the assignment problem (Def.~\ref{def:problem}) to not waste resources with no benefits.

In Figure~\ref{fig:explanation}-left, we plot the reduction of inaccessibility against $\Delta(c)$, the increase of DRT usage for each centroid $c$. $\Delta(c)$ is the amount of potential destinations for which the shortest path uses DRT, when available. As expected, the more $\Delta(c)$, the more inaccessibility is reduced.

Some features of stops $s\in\mathcal{S}^\text{DRT}$ may suggest where benefit from introducing DRT can be expected to be higher. In Figure~\ref{fig:explanation}-right, we analyze the \emph{proximity}, i.e., the number of other stops within a distance of 5km from $s$. We observe that the benefits of introducing DRT are, as expected, more pronounced when DRT serves PT stops that are close enough to other stops. In such cases indeed, passengers can arrive more easily to $s$ and from $s$ they can reach more potential destinations. Such consideration may help construct efficient heuristics to solve the assignment problem.

\begin{figure}[ht!]
\vspace{-1em}
    \centering
    \includegraphics[width=0.48\columnwidth]{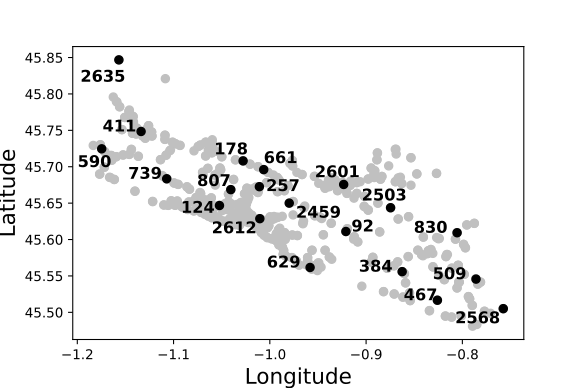}
    \setlength{\belowcaptionskip}{-13pt}
    \includegraphics[width=0.48\columnwidth]{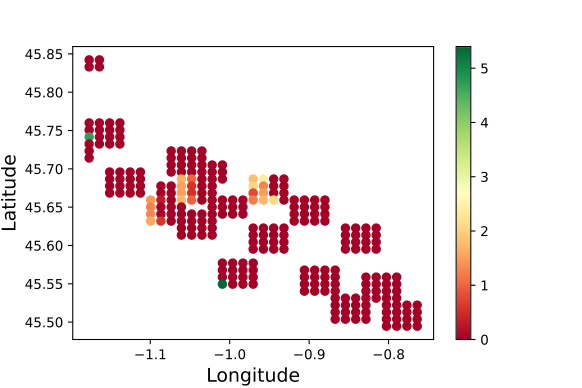}
    \vspace{-1em}
s    \caption{Left: stops $\mathcal{S}$ (in bold $\mathcal{S}^\text{DRT}$). Right: reduction in inaccessibility (log), when introducing the DRT service; $\text{inequality}_{\text{no DRT}}=\num{0.3446}$, $\text{inequality}_{\text{with DRT}}=\num{0.3443}$ (\S\ref{sec:accessibility}).}
    \label{fig:carte}
\end{figure}
\begin{figure}[ht!]
    \centering
    \includegraphics[width=0.45\columnwidth]{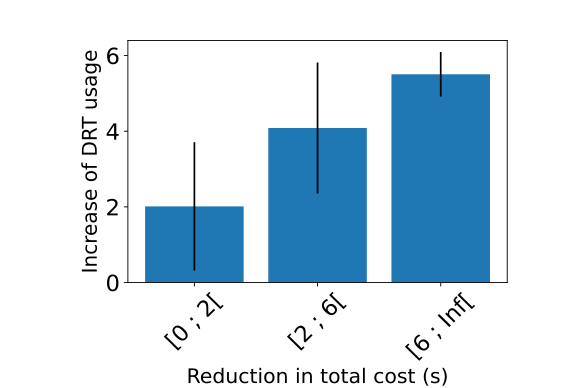}
    \includegraphics[width=0.45\columnwidth]{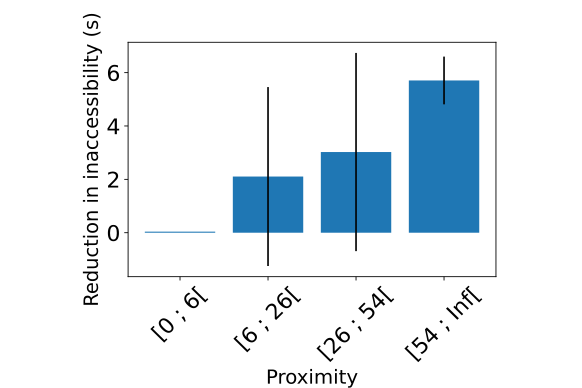}
    \vspace{-1em}
    \caption{Left: Relation between DRT usage (log) and inaccessibility reduction. Right: Relation between proximity and inaccessibility reduction (log). Centroids are partitioned in bins based on the x-axis, whose intervals are specified in the x-tics. Bars represent the average across centroids on each bin and standard deviation.}
    \label{fig:explanation}
    \vspace{-1.5em}
\end{figure}


\section{Conclusion and future work}
\label{sec:conclusion}
In this work, we set the bases of a computational approach to quantifying the impact that deploying DRT services, in a real network, has on the accessibility distribution. We show that optimally deploying DRT vehicles to this aim is NP-complete. The current work has three main limitations that we are currently addressing to make it usable for planners. First, we assume that PT demand is inelastic, as we keep demand density $\rho$ constant, while in reality the passengers choosing DRT depend on its performance that, in turn, depends on the number of passengers. To solve this circular dependency, we are currently working on solving a Traffic Assignment Problem~(TAP). Second, we are currently including in our model \emph{opportunities}, i.e., shops, jobs, schools, etc., which can be retrieved from  OpenStreetMap points of interest. We will then compute accessibility as the ease of reaching such opportunities, in accordance to the literature~\cite{Miller2019}, instead of considering the ``potential destinations'' $\mathcal{D}(c)$ (Table~\ref{tab:notation}). Third, we need to repeat our analysis in cities of different sizes to assess its scalability.

\vspace{-0.2cm}

\bibliographystyle{abbrv}
\bibliography{references}

\begin{thebibliography}{10}

\bibitem{gtfs:website}
{General Transit Feed Specification}.
\newblock \url{https://gtfs.org/}.

\bibitem{neo4j:website}
Neo4j.
\newblock \url{https://neo4j.com/}.

\bibitem{Araldo2019a}
A.~Araldo, A.~{Di Maria}, A.~{Di Stefano}, and G.~Morana.
\newblock {On the Importance of demand Consolidation in Mobility on Demand}.
\newblock In {\em IEEE/ACM DS-RT}, 2019.

\bibitem{araldo2022pooling}
A.~Araldo et~al.
\newblock Pooling for first and last mile: Integrating carpooling and transit.
\newblock In {\em hEART Conference}, 2022.

\bibitem{Badeanlou2022}
A.~Badeanlou et~al.
\newblock {Assessing transportation accessibility equity via open data}.
\newblock In {\em hEART Conference}, 2022.

\bibitem{BD18}
A.~Bonifati and S.~Dumbrava.
\newblock Graph queries: From theory to practice.
\newblock {\em {SIGMOD} Rec.}, 47(4):5--16, 2018.

\bibitem{Calabro2021}
G.~Calabrò, A.~Araldo, S.~Oh, R.~Seshadri, G.~Inturri, and M.~Ben-Akiva.
\newblock Integrating fixed and demand-responsive transportation for flexible
  transit network design.
\newblock In {\em TRB 2021: 100th Annual Meeting of the Transportation Research
  Board}, 2021.

\bibitem{Miller2019}
E.~Miller.
\newblock {Measuring Accessibility : Methods and Issues}.
\newblock In {\em International Transport Forum Roundtable on Accessibility and
  Transport Appraisal}, 2019.

\bibitem{palma2011homogeneous}
J.~G. Palma.
\newblock Homogeneous middles vs. heterogeneous tails, and the end of the
  ‘inverted-u’: It's all about the share of the rich.
\newblock {\em Development and Change}, 2011.

\bibitem{QUADRIFOGLIO2009}
L.~Quadrifoglio and X.~Li.
\newblock {A methodology to derive the critical demand density for designing
  and operating feeder transit services}.
\newblock {\em Tr. Res. P. B}, 2009.

\end{thebibliography}


\end{document}